\documentclass[12pt,twoside]{article}

 \usepackage{float}
\usepackage{graphicx}
\usepackage{epstopdf}
\usepackage{graphicx}
\usepackage{epic}
\usepackage{multirow}
\usepackage{tikz}
\usepackage{xcolor}
\usepackage{makecell}
\usepackage{threeparttable}
\usetikzlibrary{arrows,shapes,chains}

\renewcommand{\paragraph}{\roman{paragraph}}
\usepackage[a4paper]{geometry}
\makeatletter
\renewcommand\title[1]{\gdef\@title{\reset@font\Large\bfseries #1}}
\renewcommand\section{\@startsection {section}{1}{\z@}%
                                   {-3.5ex \@plus -1ex \@minus -.2ex}%
                                   {2.3ex \@plus.2ex}%
                                   {\normalfont\large\bfseries}}
\renewcommand\subsection{\@startsection{subsection}{2}{\z@}%
                                     {-3ex\@plus -1ex \@minus -.2ex}%
                                     {1.5ex \@plus .2ex}%
                                     {\normalfont\normalsize\bfseries}}
\renewcommand\subsubsection{\@startsection{subsubsection}{3}{\z@}%
                                     {-2.5ex\@plus -1ex \@minus -.2ex}%
                                     {1.5ex \@plus .2ex}%
                                     {\normalfont\normalsize\bfseries}}

\def\@runningauthor{}\newcommand{\runningauthor}[1]{\def\runningauthor{#1}}
\def\@runningtitle{}\newcommand{\runningtitle}[1]{\def\runningtitle{#1}}

\renewcommand{\ps@plain}{%
\renewcommand{\@evenhead}{\footnotesize\scshape \hfill\runningauthor\hfill}
\renewcommand{\@oddhead}{\footnotesize\scshape \hfill\runningtitle\hfill}}

\newcommand{\F}{\mathbb{F}}

\newcommand {\ccc}{{\mathbf{c}}}

\g@addto@macro\bfseries{\boldmath}

\makeatother

\usepackage{amsthm,amsmath,amssymb}
\usepackage{cite}
\usepackage{graphicx}

\usepackage[colorlinks=true,citecolor=black,linkcolor=black,urlcolor=blue]{hyperref}

\theoremstyle{plain}
\newtheorem{theorem}{Theorem}[section]

\newtheorem{lem}[theorem]{Lemma}
\newtheorem{cor}[theorem]{Corollary}
\newtheorem{prop}[theorem]{Proposition}

\theoremstyle{definition}
\newtheorem{definition}[theorem]{Definition}
\newtheorem{example}[theorem]{Example}

\theoremstyle{remark}
\newtheorem{remark}[theorem]{Remark}

\runningauthor{}

\date{}

\begin{document}
\begin{sloppypar}

\title{A further study on the mass formula for linear codes with prescribed hull dimension\thanks{The research of Shitao Li and Minjia Shi is supported by the National Natural Science Foundation of China under Grant 12071001. The research of San Ling is supported by Nanyang Technological University Research Grant 04INS000047C230GRT01.}
}
\author{Shitao Li, Minjia Shi\thanks{Shitao Li and Minjia Shi are with the Key Laboratory of Intelligent Computing and Signal Processing, Ministry of Education, School of Mathematical Sciences, Anhui University, Hefei, 230601, China (email: lishitao0216@163.com, smjwcl.good@163.com).}, Yang Li\thanks{Yang Li is with the School of Mathematics, Hefei University of Technology, Hefei, 230601, China (email: yanglimath@163.com).}, San Ling\thanks{San Ling is with the School of Physical and Mathematical Sciences, Nanyang Technological University, Singapore 637371 (email: lingsan@ntu.edu.sg).}
}

\maketitle

\begin{abstract}
Finding a mass formula for a given class of linear codes is a fundamental problem in combinatorics and coding theory.
In this paper, we consider the action of the unitary (resp. symplectic) group on the set of all Hermitian (resp. symplectic) linear complementary dual (LCD) codes, prove that all Hermitian (resp. symplectic) LCD codes are on a unique orbit under this action, and determine the formula for the size of the orbit. Based on this, we develop a general technique to obtain a closed mass formula for linear codes with prescribed Hermitian (resp. symplectic) hull dimension, and further obtain some asymptotic results.
\end{abstract}
\vspace{3mm}

\section{Introduction}

Finding a mass formula for a given class of linear codes is in fact counting the number of distinct codes in this class, and it is also a fundamental problem in combinatorics and coding theory. Mass formulas are a very important tool to classify linear codes based on some criterion (e.g. up to equivalence). Specifically, mass formulas can check whether the orbit-stabilizer theorem applied to a set of classified codes leads to the same number; see for example \cite{binaryLCD,Huffman-Kim-Sole,Huffman,2023IT-hull,Li-11,JCTA3}.
This approach can even be used to show that a set of codes found in a non-exhaustive search is complete.

The hull of a linear code over a finite field is the intersection of the code and its dual code with respect to a certain inner product (such as the Euclidean, Hermitian or symplectic inner product). The Euclidean hull was introduced by Assmus and Key \cite{A-hull-DM} to classify finite projective planes. This concept generalizes the notions of self-orthogonal codes and linear complementary dual (LCD) codes. Linear codes with prescribed hull dimension have attracted much attention lately since they have applications in lattice theory \cite{JA}, design theory \cite{JCTA1}, Boolean masking in embarked cryptographic computations \cite{lcd-appl}, the complexity of some algorithms \cite{J-p-group,Sen-p-e-codes,Sen-S-auto-group}, and quantum coding theory \cite{IT-qc,Dcc-EAqecc,QINP-sym,quantum-codes-IT-2}. There are many characterizations of linear codes with prescribed hull dimension under different inner products (see \cite{H-LCD-3,sym-LCD,LCD-Massey,Dcc-EAqecc}). However, it is a non-trivial and difficult problem to obtain their mass formulas.

There are many techniques for finding a mass formula for linear codes with prescribed Euclidean hull dimension, for example, finite geometries \cite{Pless-mass-F,Pless-mass-F2}, combinatorial numbers \cite{SIAM-DM}, group actions \cite{binaryLCD}, as well as special equations \cite{FFA-mass}. Specifically, Pless \cite{Pless-mass-F,Pless-mass-F2} in 1968 was the first one to study the mass formula and obtained a mass formula for Euclidean self-orthogonal codes by using finite geometries. Based on this, Sendrier \cite{SIAM-DM} in 1997 theoretically determined the mass formula for linear codes with prescribed Euclidean hull dimension by considering some combinatorial numbers. However, it involves the number of Euclidean self-orthogonal codes and many complex summation terms. Sendrier \cite{SIAM-DM} also believed that the mass formula he obtained is only theoretically useful and in most cases intractable. Therefore, an important topic is to obtain a simple closed mass formula for linear codes with prescribed Euclidean hull dimension.

For this topic, Carlet $et~al.$ \cite{binaryLCD} obtained a closed mass formula of $q$-ary Euclidean LCD codes by using the action of the orthogonal group on the set of all Euclidean LCD codes, where $q=2$ or $q$ is odd. Liu and Wang \cite{FFA-mass} obtained a mass formula for Euclidean LCD codes over any finite field by using a special family of generator matrices for LCD codes. Note that the mass formula obtained by Carlet $et~al.$ \cite{binaryLCD} is different in form from that obtained by Liu and Wang \cite{FFA-mass}. Li and Shi \cite{2023IT-hull} recently obtained a closed mass formula for binary linear codes with prescribed Euclidean hull dimension, and further classified (optimal) binary linear codes with small parameters. Very recently, Li $et~al.$ \cite{Li-11} presented an important characterization for the types of dual codes of linear
codes with prescribed Euclidean hull dimension. Based on this characterization, they obtained a closed mass formula for linear codes with prescribed Euclidean hull dimension over any finite field, and further classified (optimal) ternary linear codes with small parameters. Notably, these mass formulas are simpler compared with that obtained by Sendrier \cite{SIAM-DM}.

The analogous problem is also interesting and important for the Hermitian and symplectic inner products. This is because the Hermitian and symplectic inner products have greater potential than the Euclidean inner product in constructing quantum error-correcting codes (QECCs) and entanglement-assisted QECCs. Furthermore, the mass formula for linear codes with prescribed Hermitian or symplectic hull dimension is very useful for obtaining the Gilbert-Varshamov bounds of QECCs and entanglement-assisted QECCs \cite{QINP-sym,ISIT-Jin}. However, there is little research on the mass formula for linear codes with prescribed Hermitian or symplectic hull dimension. Liu and Wang \cite{FFA-mass} gave a mass formula for Hermitian LCD codes, which involves the Jacobi sum and is therefore intractable. Jin and Xing \cite{ISIT-Jin} obtained a mass formula for symplectic self-orthogonal codes in order to obtain the Gilbert-Varshamov bound of QECCs. Unfortunately, we find that there are some small typos in the formulas obtained by Jin and Xing \cite{ISIT-Jin}. Therefore, it will be necessary and interesting to conduct a more systematic study on the closed mass formula for linear codes with prescribed Hermitian or symplectic hull dimension.

In this paper, we prove that all Hermitian (resp. symplectic) LCD codes are on a unique orbit by considering the action of the unitary (resp. symplectic) group on the set of all Hermitian (resp. symplectic) LCD codes. As a consequence, we obtain mass formulas for Hermitian and symplectic LCD codes. The mass formula for Hermitian LCD codes does not involve the Jacobi sum, and hence it is simpler compared with that obtained by Liu and Wang \cite{FFA-mass}. Based on this, we develop a general technique to obtain a closed mass formula for linear codes with prescribed Hermitian (resp. symplectic) hull dimension, and further obtain some asymptotic results. We also provide corrections for the formulas obtained by Jin and Xing \cite{ISIT-Jin} on symplectic self-orthogonal codes.

The paper is organized as follows. In Section \ref{sec-2}, we give some preliminaries.
In Section \ref{sec-3}, we obtain a mass formula for linear codes with prescribed Hermitian hull dimension, and further obtain some asymptotic results. In Section \ref{sec-4}, we consider the corresponding problem for the symplectic inner product. In Section \ref{sec-5}, we conclude this paper.

\section{Preliminaries}\label{sec-2}
\subsection{Linear codes}

Throughout this paper, let $O$ denote an appropriate zero matrix and let $I_n$ denote the identity matrix of order $n$.
Let $\F_q$ denote the finite field with $q$ elements, where $q$ is a prime power. A {\em linear $[n,k]_q$ code} $C$ is a $k$-dimensional subspace of $\F_q^n$. A {\em generator matrix} of a linear $[n, k]_q$ code $C$ is any $k\times n$ matrix $G$ whose rows form a basis of $C$. For any $[n, k]_q$ linear code $C$ and $n\times n$ matrix $Q$, let $CQ$ be the linear code defined by
$$CQ=\{{\bf c}Q~|~{\bf c}\in C\}.$$

For a matrix $A$ over $\F_q$, let $A^T$ denote the transpose of $A$. For a matrix $A = (a_{ij})$ over $\F_{q^2}$, the conjugate matrix of $A$ is defined as $\overline{A} = (\overline{a_{i j}})$, where $\overline{a_{i j}}=a_{ij}^q$. Let diag$(A_1,\ldots,A_k)$ denote a block diagonal matrix, where $A_i$ is a square matrix for $1\leq i\leq k.$\vspace{2mm}

\noindent {\bf Hermitian dual} The {\em Hermitian dual} code $C^{\perp_{\rm H}}$ of a linear $[n,k]_{q^2}$ code $C$ is defined by
$$C^{\perp_{\rm H}}=\{\textbf y\in \F_{q^2}^{n}~|~\langle \textbf x, \textbf y\rangle_{\rm H}=0~ {\rm for\ all}\ \textbf x\in C \},$$
where $\langle \textbf x, \textbf y\rangle_{\rm H}=\sum_{i=1}^n x_iy_i^q$ for ${\bf x} = (x_1, x_2, \ldots, x_n)$ and $\textbf y = (y_1, y_2, \ldots, y_n)\in \F_{q^2}^{n}$. The {\em Hermitian hull} of a linear code $C$ is defined by
${\rm Hull_H}(C)= C \cap C^{\perp_{\rm H}}.$
A linear $[n,k]_{q^2}$ code $C$ is called {\em Hermitian self-orthogonal} if $C\subseteq C^{\perp_{\rm H}},$ and {\em Hermitian LCD} if $C\cap C^{\perp_{\rm H}}=\{{\bf 0}\}.$\vspace{2mm}

\noindent {\bf Symplectic dual}
For $\textbf x = (x_1, x_2,\ldots,x_{2n})$ and $\textbf y = (y_1,   y_2, \ldots,y_{2n})\in \F_q^{2n}$, we define the {\em symplectic inner product} of ${\bf x}$ and ${\bf y}$ as follows:
$$\langle {\bf x}, {\bf y}\rangle_{\rm s}=\sum_{i=1}^n(x_iy_{n+i}-x_{n+i}y_i)={\bf x}\Omega_n{\bf y}^T,~{\rm where}~\Omega_n=\begin{pmatrix}
                   O & I_n \\
                   -I_n & O
                 \end{pmatrix}.$$
The {\em symplectic dual} code $C^{\perp_{\rm s}}$ of a linear $[2n,k]_q$ code $C$ is defined by
\begin{align*}
  C^{\perp_{\rm s}}&=\{\textbf y\in \F_q^{2n}~|~\langle \textbf x, \textbf y\rangle_{\rm s}=0~ {\rm for\ all}\ \textbf x\in C \}.
\end{align*}
The {\em symplectic hull} of a linear code $C$ is defined by
${\rm Hull_s}(C)= C \cap C^{\perp_{\rm s}}.$
A linear code $C$ is called {\em symplectic self-orthogonal} if $C\subseteq C^{\perp_{\rm s}}$, and {\em symplectic LCD} if $C\cap C^{\perp_{\rm s}}=\{{\bf 0}\}.$


The following proposition will be frequently used in the paper.

\begin{prop}\label{prop-hull}
\begin{itemize}
  \item [{\rm (1)}] {\rm \cite[Proposition 3.2]{Dcc-EAqecc}} Let $C$ be a linear $[n, k]_{q^2}$ code with a generator matrix $G$. Then $C$ has $\ell$-dimensional Hermitian hull if and only if $\ell = k -{\rm rank}(G\overline{G}^T)$.
  \item [{\rm (2)}] {\rm \cite[Proposition 2]{sym-SO}} Let $C$ be a linear $[2n, k]_{q}$ code with a generator matrix $G$. Then $C$ has $\ell$-dimensional symplectic hull if and only if $\ell = k -{\rm rank}(G\Omega_nG^T)$.
\end{itemize}
\end{prop}

For a number $q$ with $q\neq 1$, the Gaussian binomial coefficient $\begin{bmatrix}
                     n \\
                     k
                   \end{bmatrix}_q$ is defined to be
$$
 \begin{bmatrix}
                    n \\
                    k
                  \end{bmatrix}_q=
 \frac{(q^n-1)(q^n-q)\cdots(q^n-q^{k-1})}{(q^k-1)(q^k-q)\cdots(q^k-q^{k-1})}~{\rm if}~k\neq 0~{\rm and}~\begin{bmatrix}
                    n \\
                    0
                  \end{bmatrix}_q=1.$$
The Gaussian binomial coefficients have the same symmetry as that of
binomial coefficients, $i.e.$,
$\begin{bmatrix}
                     n \\
                     k
                   \end{bmatrix}_q=\begin{bmatrix}
                     n \\
                     n-k
                   \end{bmatrix}_q$.

\subsection{Unitary group and symplectic group}
Let $\mathbb{G}\mathbb{L}_{n}(q)$ denote the {\em general linear group} of degree $n$ over $\F_q$, which is the set of $n\times n$ invertible matrices over $\F_q$.
An $n\times n$ matrix $Q$ over $\F_{q^2}$ is called a {\em unitary matrix} if $Q\in \mathbb{G}\mathbb{L}_{n}(q^2)$ and $Q\overline{Q}^T= I_n$. It can be checked that an $n\times n$ matrix $Q$ is a unitary matrix if and only if
$\langle{\bf u}Q,{\bf v}Q\rangle_{\rm H}=\langle{\bf u},{\bf v}\rangle_{\rm H}$ for all ${\bf u},{\bf v}\in \F_{q^2}^{n}$. The {\em unitary group} is defined by
$\mathbb{U}_{n}(q^2)=\{Q\in \mathbb{G}\mathbb{L}_{n}(q^2)~|~Q\overline{Q}^T=I_n\}.$
The size of the unitary group can be found in \cite[Theorem 11.28]{sym-group}, which says that
\begin{align}\label{eq-unitary-g}
|\mathbb{U}_{n}(q^2)|&=q^{\frac{n(n-1)}{2}}\prod_{i=1}^{n}(q^i-(-1)^i).
\end{align}

A $2n\times 2n$ matrix $Q$ is called a {\em symplectic matrix} with respect to $\Omega_n$ if $Q\Omega_nQ^T=\Omega_n$. It can be checked that a $2n\times 2n$ matrix $Q$ is a symplectic matrix with respect to $\Omega_n$ if and only if
$\langle{\bf u}Q,{\bf v}Q\rangle_{\rm s}=\langle{\bf u},{\bf v}\rangle_{\rm s}$ for all ${\bf u},{\bf v}\in \F_q^{2n}$.
The {\em symplectic group} with respect to $\Omega_n$ is defined by
$$\mathbb{S}\mathrm{p}^1_{2n}(q)=\{Q\in \mathbb{G}\mathbb{L}_{2n}(q)~|~Q\Omega_nQ^T=\Omega_n\}.$$
Similarly, the {\em symplectic group} with respect to ${\rm diag}(\underbrace{J_{2},J_{2},\ldots,J_{2}}_{n})$ is defined by
$$\mathbb{S}\mathrm{p}^2_{2n}(q)=\{Q\in \mathbb{G}\mathbb{L}_{2n}(q)~|~Q{\rm diag}(\underbrace{J_{2},J_{2},\ldots,J_{2}}_{n})Q^T={\rm diag}(\underbrace{J_{2},J_{2},\ldots,J_{2}}_{n})\},$$
where $J_2=\begin{pmatrix}
               0 & 1 \\
               -1 & 0 \\
             \end{pmatrix}.$
It is easy to see that there is a one-to-one correspondence between $\mathbb{S}\mathrm{p}^1_{2n}(q)$ and $\mathbb{S}\mathrm{p}^2_{2n}(q)$.
The size of the symplectic group can be found in \cite{sym-group}, which says that
\begin{align}\label{eq-sp}
|\mathbb{S}\mathrm{p}^1_{2n}(q)|&=|\mathbb{S}\mathrm{p}^2_{2n}(q)|
=q^{n^2}\prod_{i=1}^n(q^{2i}-1).
\end{align}

\section{A mass formula for linear codes with prescribed Hermitian hull dimension}\label{sec-3}

The main purpose of this section is to obtain a mass formula for linear codes with prescribed Hermitian hull dimension. We first recall a characterization of linear codes with prescribed Hermitian hull dimension as follows.

\begin{prop}{\rm\cite[Theorem 14]{FFA-mass}}\label{prop-character-Hhull}
Let $C$ be a linear $[n,k]_{q^2}$ code. Then $C$ has $\ell$-dimensional Hermitian hull if and only if there exists a generator matrix $G$ of $C$ such that $G\overline{G}^T={\rm diag}(\underbrace{1,1,\ldots,1}_{k-\ell},0,\ldots,0).$
\end{prop}

\subsection{A mass formula for Hermitian LCD codes}

In this subsection, we obtain a closed mass formula for Hermitian LCD codes. Let HLCD$[n,k]_{q^2}$ denote the set of all Hermitian LCD $[n,k]_{q^2}$ codes. Let ${\bf e}_i=(0,\ldots,0,1,0,\ldots,0)$ be the vector with 1 in the $i$th position. Let $G^0_{n,k,q^2}$ and $H^0_{n,k,q^2}$ denote respectively $k\times n$ and $(n-k)\times n$ matrices defined by
\begin{align}\label{Eq-G0-H0}
  G^0_{n,k,q^2}=\left(\begin{array}{c}
             {\bf e}_1 \\
             {\bf e}_2 \\
             \vdots \\
             {\bf e}_k
           \end{array}
\right) &~{\rm and}~H^0_{n,k,q^2}=\left(\begin{array}{c}
             {\bf e}_{k+1} \\
             {\bf e}_{k+2}\\
             \vdots \\
             {\bf e}_n
           \end{array}
\right)
\end{align}
Let $C^0_{n,k,q^2}$ denote the linear code generated by the matrix $G^0_{n,k,q^2}$. Then it can be checked that $C^0_{n,k,q^2}\in {\rm HLCD}[n,k]_{q^2}$ and $H^0_{n,k,q^2}$ is a generator matrix of $(C^0_{n,k,q^2})^{\perp_{\rm H}}$.

\begin{theorem}\label{thm-C1-U-C2}
Let $C_1$ and $C_2$ be two Hermitian LCD $[n,k]_{q^2}$ codes. Then there exists a unitary matrix $Q\in \mathbb{U}_{n}(q^2)$ such that $C_2=C_1Q$. Conversely, for any Hermitian LCD $[n,k]_{q^2}$ code $C$ and any unitary matrix $Q\in \mathbb{U}_{n}(q^2)$, $CQ$ is also a Hermitian LCD $[n,k]_{q^2}$ code.
\end{theorem}

\begin{proof}
By Proposition \ref{prop-character-Hhull}, there are generator matrices $G_1$ and $H_1$ of $C_1$ and $C_1^{\perp_{\rm H}}$, respectively, such that $G_1\overline{G_1}^T=I_k$ and $H_1\overline{H_1}^T=I_{n-k}$. Similarly, there are generator matrices $G_2$ and $H_2$ of $C_2$ and $C_2^{\perp_{\rm H}}$, respectively, such that $G_2\overline{G_2}^T=I_k$ and $H_2\overline{H_2}^T=I_{n-k}$. Let $Q_1$ and $Q_2$ be the matrices defined by
$
Q_1=\begin{pmatrix}
G_1   \\
H_1   \\
\end{pmatrix}~{\rm and}~Q_2=
\begin{pmatrix}
G_2    \\
H_2    \\
\end{pmatrix}.
$
Then $Q_1$ and $Q_2$ are nonsingular and $Q_1\overline{Q_1}^T=Q_2\overline{Q_2}^T=I_n.$
Let $Q=Q_1^{-1}Q_2$. Then $Q\overline{Q}^T=(Q_1^{-1}Q_2)(\overline{Q_1^{-1}Q_2})^T=I_n$. Further, $G_1Q=G_1Q_1^{-1}Q_2=G_1\overline{Q_1}^TQ_2=G_1(\overline{G_1}^T ~ \overline{H_1}^T)\begin{pmatrix}
G_2    \\
H_2    \\
\end{pmatrix}=(I_k~O)\begin{pmatrix}
G_2    \\
H_2    \\
\end{pmatrix}=G_2$. Hence we have $Q\in \mathbb{U}_n(q^2)$ and $C_1Q=C_2$.

Conversely, let $C$ be a Hermitian LCD $[n,k]_{q^2}$ code and let $Q\in \mathbb{U}_{n}(q^2)$ be a unitary matrix. Let $G$ be a generator matrix $C$. Then $GQ$ is a generator matrix of $CQ$. Since
$(GQ)(\overline{GQ})^T=G(Q\overline{Q}^T)\overline{G}^T=G\overline{G}^T$, $CQ$ is also a Hermitian LCD $[n,k]_{q^2}$ code.
\end{proof}

By Theorem \ref{thm-C1-U-C2}, the unitary group $\mathbb{U}_{n}(q^2)$ acts on the set ${\rm HLCD}[n,k]_{q^2}$ by $(C,Q)\longmapsto CQ$, where $C\in {\rm HLCD}[n,k]_{q^2}$ and $Q\in \mathbb{U}_{n}(q^2)$. Under this action, all Hermitian LCD $[n,k]_{q^2}$ codes are on a unique orbit. Note that the code $C^0_{n,k,q^2}$ generated by the matrix $G^0_{n,k,q^2}$ defined in Equation (\ref{Eq-G0-H0}) is a Hermitian LCD $[n,k]_{q^2}$ code. Hence we have
\begin{align}\label{Eq-HLCD}
  {\rm HLCD}[n,k]_{q^2}& = C^0_{n,k,q^2}\mathbb{U}_n(q^2)=\{C^0_{n,k,q^2}Q~|~Q\in \mathbb{U}_n(q^2)\}.
\end{align}
Hence, to determine the size of ${\rm HLCD}[n,k]_{q^2}$, it is sufficient to determine the size of the stabilizer of $C^0_{n,k,q^2}$ in the unitary group. The {\em stabilizer} of a linear $[n,k]_{q^2}$ code $C$ in the unitary group is defined by
$$\mathbb{S}{\rm t}_u(C)=\{Q\in \mathbb{U}_{n}(q^2)~|~CQ=C\}.$$

\begin{prop}\label{prop-St}
Let $C$ be a Hermitian LCD $[n,k]_{q^2}$ code. Let $G$ be a generator matrix of $C$ and let $H$ be a generator matrix of $C^{\perp_{\rm H}}$. Then $Q\in \mathbb{S}{\rm t}_u(C)$ if and only if
$$Q=\begin{pmatrix}
 G \\
 H
 \end{pmatrix}^{-1}
 \begin{pmatrix}
 Q_1&O \\
 O& Q_2
 \end{pmatrix}
 \begin{pmatrix}
 G \\
 H
 \end{pmatrix},$$
 where $Q_1\in \mathbb{G}\mathbb{L}_{k}(q)$, $Q_2\in \mathbb{G}\mathbb{L}_{n-k}(q)$ such that $Q_1(G\overline{G}^T)\overline{Q_1}^T=G\overline{G}^T$ and $Q_2(H\overline{H}^T)\overline{Q_2}^T=H\overline{H}^T$.
\end{prop}

\begin{proof}
Let $Q=\begin{pmatrix}
 G \\
 H
 \end{pmatrix}^{-1}
 \begin{pmatrix}
 Q_1&O \\
 O& Q_2
 \end{pmatrix}
 \begin{pmatrix}
 G \\
 H
 \end{pmatrix},$
 where $Q_1\in \mathbb{G}\mathbb{L}_{k}(q)$, $Q_2\in \mathbb{G}\mathbb{L}_{n-k}(q)$ such that $Q_1(G\overline{G}^T)\overline{Q_1}^T=G\overline{G}^T$ and $Q_2(H\overline{H}^T)\overline{Q_2}^T=H\overline{H}^T$. Then
 $$\begin{pmatrix}
 G \\
 H
 \end{pmatrix}Q=
 \begin{pmatrix}
 Q_1&O \\
 O& Q_2
 \end{pmatrix}
 \begin{pmatrix}
 G \\
 H
 \end{pmatrix}=\begin{pmatrix}
 Q_1G \\
 Q_2H
 \end{pmatrix}.$$
 Hence $GQ=Q_1G$ and $HQ=Q_2H$, which imply that $CQ=C$. Note that
\begin{align*}
  \begin{pmatrix}
 G \\
 H
 \end{pmatrix}Q\overline{Q}^T\begin{pmatrix}
 \overline{G} \\
 \overline{H}
 \end{pmatrix}^T & =\left(\begin{pmatrix}
 G \\
 H
 \end{pmatrix}Q\right)\left(\begin{pmatrix}
 \overline{G} \\
 \overline{H}
 \end{pmatrix}\overline{Q}\right)^T \\
   & =\begin{pmatrix}
 Q_1G \\
 Q_2H
 \end{pmatrix}\begin{pmatrix}
 \overline{Q_1G} \\
 \overline{Q_2H}
 \end{pmatrix}^T\\
 &=\begin{pmatrix}
 Q_1(G\overline{G}^T)\overline{Q_1}^T &O\\
 O& Q_2(H\overline{H}^T)\overline{Q_2}^T
 \end{pmatrix}\\
 &=\begin{pmatrix}
G\overline{G}^T &O\\
 O& H\overline{H}^T
 \end{pmatrix}\\
 &= \begin{pmatrix}
 G \\
 H
 \end{pmatrix}\begin{pmatrix}
 \overline{G} \\
 \overline{H}
 \end{pmatrix}^T.
\end{align*}
It follows that
$$Q\overline{Q}^T=\begin{pmatrix}
 G \\
 H
 \end{pmatrix}^{-1}\left(\begin{pmatrix}
 G \\
 H
 \end{pmatrix}\begin{pmatrix}
 \overline{G} \\
 \overline{H}
 \end{pmatrix}^T\right)\begin{pmatrix}
 \overline{G} \\
 \overline{H}
 \end{pmatrix}^{-T}=I_n,$$
which implies that $Q\in \mathbb{U}_{n}(q^2)$. Thus, $Q\in \mathbb{S}{\rm t}_u(C)$.

Conversely, let $Q\in \mathbb{S}{\rm t}_u(C)$, $i.e.$, $Q\in \mathbb{U}_{n}(q^2)$ and $CQ=C$. Then there exists a matrix $Q_1\in \mathbb{G}\mathbb{L}_{n}(q)$ such that
$GQ=Q_1G$. For any ${\bf c}\in C$ and any ${\bf w}\in C^{\perp_{\rm H}}$, we have
$\langle{\bf c}Q,{\bf w}Q\rangle_{\rm H}=\langle{\bf c},{\bf w}\rangle_{\rm H}.$
Then $C^{\perp_{\rm H}}Q=C^{\perp_{\rm H}}$. Hence there exists a matrix $Q_2\in \mathbb{G}\mathbb{L}_{n}(q)$ such that $HQ=Q_2H$. It follows that
$$\begin{pmatrix}
 G \\
 H
 \end{pmatrix}Q=\begin{pmatrix}
 Q_1G \\
 Q_2H
 \end{pmatrix}=
 \begin{pmatrix}
 Q_1&O \\
 O& Q_2
 \end{pmatrix}
 \begin{pmatrix}
 G \\
 H
 \end{pmatrix},~i.e.,~Q=\begin{pmatrix}
 G \\
 H
 \end{pmatrix}^{-1}
 \begin{pmatrix}
 Q_1&O \\
 O& Q_2
 \end{pmatrix}
 \begin{pmatrix}
 G \\
 H
 \end{pmatrix}.$$
Since $Q\in \mathbb{U}_{n}(q^2)$, $Q\overline{Q}=I_n$. Hence we have
\begin{align*}
  G\overline{G}^T&=G(Q\overline{Q}^T)\overline{G}^T=(GQ)(\overline{GQ})^T=
(Q_1G)(\overline{Q_1G})^T=Q_1(G\overline{G}^T)\overline{Q_1}^T,\\
  H\overline{H}^T&=H(Q\overline{Q}^T)\overline{H}^T=(HQ)(\overline{HQ})^T=
(Q_2H)(\overline{Q_2H})^T=Q_2(H\overline{H}^T)\overline{Q_2}^T.
\end{align*}
This completes the proof.
\end{proof}

\begin{theorem}\label{thm-mass-HHull}
The number of Hermitian LCD $[n,k]_{q^2}$ codes is
$$|{\rm HLCD}[n,k]_{q^2}|=q^{k(n-k)}\prod_{i=1}^{k}\frac{q^{n-k+i}-(-1)^{n-k+i}} {q^i-(-1)^i}.$$
\end{theorem}

\begin{proof}
Let $C^0_{n,k,q^2}$ be the Hermitian LCD $[n,k]_{q^2}$ code with the generator matrix $G^0_{n,k,q^2}$ defined in Equation (\ref{Eq-G0-H0}). Then the matrix $H^0_{n,k,q^2}$ defined in Equation (\ref{Eq-G0-H0}) is a generator matrix of $(C^0_{n,k,q^2})^{\perp_{\rm H}}$. Note that $G^0_{n,k,q^2}\overline{G^0_{n,k,q^2}}^T=I_k$ and $H^0_{n,k,q^2}\overline{H^0_{n,k,q^2}}^T=I_{n-k}$. Together with Proposition \ref{prop-St} and the fact that $\begin{pmatrix}
 G^0_{n,k,q^2} \\
 H^0_{n,k,q^2}
 \end{pmatrix}=\begin{pmatrix}
 G^0_{n,k,q^2} \\
 H^0_{n,k,q^2}
 \end{pmatrix}^{-1}=I_n$, we have
$$\mathbb{S}{\rm t}_u(C^0_{n,k,q^2})=\left\{\begin{pmatrix}
 Q_1&O \\
 O& Q_2
 \end{pmatrix}~\bigg |~ Q_1\in \mathbb{U}_k(q^2),~Q_2\in \mathbb{U}_{n-k}(q^2)\right\}.$$
From Equation (\ref{Eq-HLCD}), we have
$$|{\rm HLCD}[n,k]_{q^2}|=\frac{|\mathbb{U}_{n}(q^2)|}{|\mathbb{S}{\rm t}_u(C^0_{n,k,q^2})|}=\frac{|\mathbb{U}_{n}(q^2)|}{|\mathbb{U}_{k}(q^2)|\cdot |\mathbb{U}_{n-k}(q^2)|}.$$
Combining with Equation (\ref{eq-unitary-g}), we complete the proof.
\end{proof}

\begin{remark}
Liu and Wang {\rm\cite{FFA-mass}} gave a mass formula for Hermitian LCD codes, which involves the Jacobi sum and is therefore intractable. With respect to their result, the mass formula we give here has an obviously simpler form.
\end{remark}

\subsection{A mass formula for linear codes with prescribed Hermitian hull dimension}

A vector ${\bf x}=(x_1,x_2,\ldots,x_n)\in \F_{q^2}^n$ is {\em Hermitian self-orthogonal} if $\langle{\bf x},{\bf x}\rangle_{\rm H}=\sum_{i=1}^n x_i^{q+1}=0$. Let ${\rm HHull}[n,k]^{\ell}_{q^2}$ denote the set of all linear $[n,k]_{q^2}$ codes $C$ with $\ell$-dimensional Hermitian hull.
In this subsection, our aim is to obtain a closed mass formula for linear codes with prescribed Hermitian hull dimension. To do this, we give a crucial lemma.
\begin{lem}\label{lem-solution}
Let $a\in \F_q$ and let $N^{q+1}_a(n)$ denote the number of solutions of the diagonal equation $x_1^{q+1}+x_2^{q+1}+\cdots+x_n^{q+1}=a$ in $\F_{q^2}^n$. If $a=0$, then
$$N^{q+1}_0(n)=\left\{\begin{array}{ll}
                       q^{n-1}(q^n-q+1), & {\rm if}~n~{\rm is~odd}, \\
                       q^{n-1}(q^n+q-1), & {\rm if}~n~{\rm is~even}.
                     \end{array}
\right.$$
If $a\neq 0$, then
$$N^{q+1}_a(n)=\left\{\begin{array}{ll}
                       q^{n-1}(q^n+1), & {\rm if}~n~{\rm is~odd}, \\
                       q^{n-1}(q^n-1), & {\rm if}~n~{\rm is~even}.
                     \end{array}
\right.$$
\end{lem}

\begin{proof}
It is well-known that $x^{q+1}\in \F_q$ for any $x\in \F_{q^2}$ and $|\{x^{q+1}=a~|~x\in \F_{q^2}\}|=q+1$ for any $a\in \F_q^*$.
For any $a\in \F_{q}^*$, there exists $c\in \F_{q^2}^*$ such that $a=c^{q+1}$. Hence
$x_1^{q+1}+x_2^{q+1}+\cdots+x_n^{q+1}=a$ is equivalent to
$\left(\frac{x_1}{c}\right)^{q+1}+\left(\frac{x_2}{c}\right)^{q+1}+\cdots+\left(\frac{x_n}{c}\right)^{q+1}=1,$
and hence $N^{q+1}_a(n)=N^{q+1}_1(n)$ for any $a\in \F_{q}^*$.

Assume that $(y_1,\ldots,y_n)$ is a solution of the diagonal equation $x_1^{q+1}+\cdots+x_n^{q+1}=a$ in $\F_{q^2}^n$, where $a\in \F_q$.
\begin{itemize}
  \item For fixed $y_n\in \F_{q^2}$ with $y_n^{q+1}=a$, the number of $(y_1,\ldots,y_{n-1})\in \F_{q^2}^{n-1}$ satisfying $y_1^{q+1}+\cdots+y_n^{q+1}=a$ is equal to the number of $(y_1,\ldots,y_{n-1})\in \F_{q^2}^{n-1}$ satisfying $y_1^{q+1}+\cdots+y_{n-1}^{q+1}=0$, $i.e.$, $N^{q+1}_0(n-1)$.
  \item For fixed $y_n\in \F_{q^2}$ with $y_n^{q+1}\neq a$, the number of $(y_1,\ldots,y_{n-1})\in \F_{q^2}^{n-1}$ satisfying $y_1^{q+1}+\cdots+y_n^{q+1}=a$ is equal to the number of $(y_1,\ldots,y_{n-1})\in \F_{q^2}^{n-1}$ satisfying $y_1^{q+1}+\cdots+y_{n-1}^{q+1}=a-y_n^{q+1}\neq0$, $i.e.$, $N^{q+1}_1(n-1)$.
\end{itemize}
Hence it can be checked that
\begin{align*}
  N^{q+1}_0(n)&=N^{q+1}_0(n-1)+(q^2-1)N^{q+1}_1(n-1),~{\rm and} \\
  N^{q+1}_1(n)&=(q+1)N^{q+1}_0(n-1)+(q^2-q-1)N^{q+1}_1(n-1).
\end{align*}

Next, use induction for $n$. If $n=1$, then the result is obviously valid.
Assume that the result is valid for $n=t$. Next we prove that the result is also valid for $n=t+1$. If $t$ is odd, then
\begin{align*}
  N^{q+1}_0(t+1) & =N^{q+1}_0(t)+(q^2-1)N^{q+1}_1(t) \\
   &=q^{t-1}(q^{t}-q+1)+(q^2-1)q^{t-1}(q^{t}+1)\\
   &=q^t(q^{t+1}+q-1)
\end{align*}
and
\begin{align*}
  N^{q+1}_1(t+1) & =(q+1)N^{q+1}_0(t)+(q^2-q-1)N^{q+1}_1(t)\\
   & =(q+1)q^{t-1}(q^{t}-q+1)+(q^2-q-1)q^{t-1}(q^{t}+1)\\
   &=q^t(q^{t+1}-1).
\end{align*}
The case where $t$ is even is similar to the case where $t$ is odd. Hence the result is also valid for $n=t+1$. By induction hypothesis, the result follows.
\end{proof}

\begin{prop}\label{prop-HSO}
Let $C$ be a Hermitian LCD $[n,k]_{q^2}$ code. Then the number of Hermitian self-orthogonal codewords of $C$ is $N_0^{q+1}(k)$, where $N_0^{q+1}(k)$ is defined in Lemma $\ref{lem-solution}$.
\end{prop}

\begin{proof}
By Equation (\ref{Eq-HLCD}), there exists $Q\in \mathbb{U}_n(q^2)$ such that $C=C^0_{n,k,q^2}Q$, where $C^0_{n,k,q^2}$ is the Hermitian LCD $[n,k]_{q^2}$ code with the generator matrix $G^0_{n,k,q^2}$ defined by Equation (\ref{Eq-G0-H0}). Since $\langle {\bf u}Q,{\bf v}Q\rangle_{\rm H}=\langle {\bf u},{\bf v}\rangle_{\rm H}$, the number of Hermitian self-orthogonal codewords ${\bf c}\in C$ is the same as that of $C^0_{n,k,q^2}$. By Lemma \ref{lem-solution}, the number of Hermitian self-orthogonal codewords of $C^0_{n,k,q^2}$ is $N^{q+1}_0(k)$. This completes the proof.
\end{proof}

It is well-known that the number of solutions of a diagonal equation can be expressed in term of Jacobi sums \cite[Ch. 6]{finite-field}. As a result of Lemma \ref{lem-solution}, we obtain the exact values of two summation terms related to Jacobi sums.

\begin{definition}
Let $\lambda_1,\ldots,\lambda_k$ be $k$ multiplicative characters of $\F_{q}$ and let $a\in \F_q$ be fixed. Then the {\em Jacobi sum} in $\F_q$ is defined by
$$J_a(\lambda_1,\ldots,\lambda_k)=\sum_{c_1+\cdots+c_k=a}\lambda_1(c_1) \cdots\lambda_k(c_k),$$
where the summation is extended over all $k$-tuples $(c_1,\ldots,c_k)$ of elements of $\F_q$ with $c_1+\cdots+c_k=a$.
\end{definition}

\begin{cor}{\rm\cite[Theorem 6.33]{finite-field}}
Let $\lambda$ be a multiplicative character of $\F_{q^2}$ of order $q+1$. Then we have
\begin{align*}
  q^{2n-2}+\sum_{{\tiny\begin{array}{c}
1\leq j_1,\ldots,j_{n}\leq q \\
(q+1)|j_1+\cdots+j_{n}
\end{array}}}J_0(\lambda^{j_1},\ldots,\lambda^{j_{n}})
&=\left\{\begin{array}{ll}\vspace{2mm}
           q^{n-1}(q^n-q+1), & n~{\rm is~odd}, \\
           q^{n-1}(q^n+q-1), & n~{\rm is~even}.
         \end{array}
\right. \\
  q^{2n-2}+\sum_{j_1=1}^q\sum_{j_2=1}^q\cdots \sum_{j_n=1}^q J_1(\lambda^{j_1},\ldots,\lambda^{j_{n}})
&=\left\{\begin{array}{ll}\vspace{2mm}
           q^{n-1}(q^n+1), & n~{\rm is~odd}, \\
           q^{n-1}(q^n-1), & n~{\rm is~even}.
         \end{array}
\right.
\end{align*}
\end{cor}

\begin{proof}
By \cite[Theorem 6.33]{finite-field}, the number of solutions of the diagonal equation $x_1^{q+1}+x_2^{q+1}+\cdots+x_n^{q+1}=0$ in $\F_{q^2}^n$ is
$$q^{2n-2}+\sum_{{\tiny\begin{array}{c}
1\leq j_1,\ldots,j_{n}\leq q \\
(q+1)|j_1+\cdots+j_{n}
\end{array}}}J_0(\lambda^{j_1},\ldots,\lambda^{j_{n}}).$$
By \cite[Theorem 6.34]{finite-field}, the number of solutions of the diagonal equation $x_1^{q+1}+x_2^{q+1}+\cdots+x_n^{q+1}=1$ in $\F_{q^2}^n$ is
$$q^{2n-2}+\sum_{j_1=1}^q\sum_{j_2=1}^q\cdots \sum_{j_n=1}^q J_1(\lambda^{j_1},\ldots,\lambda^{j_{n}}).$$
Combining with Lemma \ref{lem-solution}, we complete the proof.
\end{proof}

\begin{prop}\label{Prop-HHull}
Let $C$ be a linear $[n,k]_{q^2}$ code with $\ell$-dimensional Hermitian hull. Then the number of Hermitian self-orthogonal codewords of $C\backslash {\rm Hull_H}(C)$ is
$\left(N_0^{q+1}(k-\ell)-1\right)q^{2\ell}$, where $N_0^{q+1}(k)$ is defined in Lemma $\ref{lem-solution}$.
\end{prop}

\begin{proof}
For any linear $[n,k]_{q^2}$ code $C$ with $\ell$-dimensional Hermitian hull, it follows from \cite[Theorem 2.10]{LCD-code-Li} that $C$ is the direct sum of the Hermitian self-orthogonal $[n,\ell]_{q^2}$ code $C_1={\rm Hull_H}(C)$ and a Hermitian LCD $[n,k-\ell]_{q^2}$ code $C_2$. For any ${\bf c}\in C\backslash {\rm Hull_H}(C)$, there exists ${\bf c}_1\in C_1$ and ${\bf c}_2\in C_2$ such that ${\bf c}={\bf c}_1+{\bf c}_2$, where ${\bf c}_2$ is a nonzero codeword.
Hence
$\langle {\bf c},{\bf c}\rangle_{\rm H}=\langle {\bf c}_1+{\bf c}_2,{\bf c}_1+{\bf c}_2\rangle_{\rm H}=\langle {\bf c}_2,{\bf c}_2\rangle_{\rm H}.$
This implies that ${\bf c}$ is Hermitian self-orthogonal if and only if ${\bf c}_2$ is Hermitian self-orthogonal. By Proposition \ref{prop-HSO}, the number of nonzero Hermitian self-orthogonal codewords of $C_2$ is $N_0^{q+1}(k-\ell)-1$. Combining with the fact that $C_1$ has $q^{2\ell}$ codewords, we complete the proof.
\end{proof}

\begin{theorem}\label{Thm-mass-Hhull}
Let $n,k,\ell$ be three positive integers such that $\ell\leq k\leq n-\ell$. Let $k_0=k-\ell$. Then we have
$$|{\rm HHull}[n,k]^\ell_{q^2}|=\left\{\begin{array}{ll}\vspace{0.2cm}
\left(\prod_{i=1}^\ell A_{n,k_0,i}\right) q^{k_0(n-k_0)}\prod_{i=1}^{k_0}\frac{q^{n-k_0+i}-(-1)^{n-k_0+i}} {q^i-(-1)^i}, &  {\rm if}\ n-k_0~{\rm is~odd}, \\
\left(\prod_{i=1}^\ell B_{n,k_0,i}\right) q^{k_0(n-k_0)}\prod_{i=1}^{k_0}\frac{q^{n-k_0+i}-(-1)^{n-k_0+i}} {q^i-(-1)^i}, &  {\rm if}\ n-k_0~{\rm is~even},
\end{array}
\right.$$
where $A_{n,k_0,i}=\frac{(q^{n-k_0-2i+2}+1)(q^{n-k_0-2i+1}-1)}{q^{2k_0}(q^{2i}-1)}$ and $B_{n,k_0,i}=\frac{(q^{n-k_0-2i+2}-1)(q^{n-k_0-2i+1}+1)}{q^{2k_0}(q^{2i}-1)}$.
\end{theorem}

\begin{proof}
For convenience, we use $C_{n,k}^\ell$ to denote a linear $[n, k]_{q^2}$ code with $\ell$-dimensional Hermitian hull. Then a given $C_{n,k-1}^{\ell-1}$ can be extended to a $C_{n,k}^\ell$ by adding a Hermitian self-orthogonal codeword of $(C_{n,k-1}^{\ell-1})^{\perp_{\rm H}}\backslash {\rm Hull_H}(C_{n,k-1}^{\ell-1})$.
By Proposition \ref{Prop-HHull}, one can construct $\left(N_0^{q+1}(n-k-\ell+2)-1\right)q^{2(\ell-1)}$ linear $[n,k]_{q^2}$ codes with $\ell$-dimensional Hermitian hull from $C_{n,k-1}^{\ell-1}$ through the above way. Note that some of these codes may be the same.
On the other hand, for any $C^{\ell}_{n,k}$, there are $(q^{2\ell}-1)q^{2(k-1)}$ ways to obtain $C^{\ell}_{n,k}$ through the above way by a similar discussion as in \cite[Proposition 5.5]{2023IT-hull}. Hence we have
\begin{align*}
  |{\rm HHull}[n,k]_{q^2}^\ell|&=
  \frac{\left(N_0^{q+1}(n-k-\ell+2)-1\right)q^{2(\ell-1)}} {(q^{2\ell}-1)q^{2(k-1)}}|{\rm HHull}[n,k-1]^{\ell-1}_{q^2}|\\
  &=\frac{N_0^{q+1}(n-k_0-2\ell+2)-1} {q^{2k}-q^{2k_0}}|{\rm HHull}[n,k-1]^{\ell-1}_{q^2}|
\end{align*}
Repeating the above operation $\ell$ times, we obtain
\begin{align*}
  |{\rm HHull}[n,k]^\ell_{q^2}|&=\left(\prod_{i=1}^\ell \frac{N_0^{q+1}(n-k_0-2i+2)-1} {q^{2k_0+2i}-q^{2k_0}}\right)|{\rm HHull}[n,k_0]^0_{q^2}|.
\end{align*}
Note that ${\rm HHull}[n,k_0]^0_{q^2}={\rm HLCD}[n,k_0]_{q^2}$ and
$$N_0^{q+1}(n-k_0-2i+2)-1=\left\{\begin{array}{ll}\vspace{0.2cm}
(q^{n-k_0-2i+2}+1)(q^{n-k_0-2i+1}-1), &  {\rm if}\ n-k_0~{\rm is~odd}, \\
(q^{n-k_0-2i+2}-1)(q^{n-k_0-2i+1}+1), &  {\rm if}\ n-k_0~{\rm is~even}.
\end{array}
\right.$$
 Hence the result holds by Theorem \ref{thm-mass-HHull}.
\end{proof}

\begin{cor}
The number of Hermitian self-orthogonal $[n,k]_{q^2}$ codes is
$$\prod_{i=1}^k\frac{q^{n-2i+1}(q^{n-2i+2}+(-1)^n(q-1))-1}{q^{2i}-1}.$$
\end{cor}

\begin{proof}
Note that a Hermitian self-orthogonal $[n,k]_{q^2}$ code is a linear $[n,k]_{q^2}$ code with $k$-dimensional Hermitian hull. This yields that $k_0=k-\ell=0$. Hence by Theorem \ref{Thm-mass-Hhull}, the desired result holds.
\end{proof}

In the following, we count the inequivalent linear codes with prescribed Hermitian hull dimension. This also provides a detailed explanation for the first paragraph of the introduction. Let $\mathbb{P}_n$ denote the group generated by all $n\times n$ permutation matrices, i.e., square matrices that have exactly one entry equal to $1$ in each row and each column and $0$'s elsewhere. Two linear codes $C_1,C_2$ are {\em permutation-equivalent} if there exists a permutation $P\in \mathbb{P}_n$ such that $C_2=C_1P=\{{\bf c}P~|~{\bf c}\in C_1\}$. Note that $C$ and $CP$ are in the same orbit for any $C\in {\rm HHull}[n,k]_{q^2}^\ell$ and any $P\in \mathbb{P}_n$. For any $[n, k]_{q^2}$ linear code $C$, the automorphism group Aut$(C)$ of $C$ is defined by $\{P\in \mathbb{P}_n~|~CP=C\}.$ For any $P_0\in \mathbb{P}_n$, it can be checked that
$${\rm Aut}(CP_0)=P_0{\rm Aut}(C)=\{P_0P~|~P\in {\rm Aut}(C)\}.$$
Hence $|{\rm Aut}(CP_0)|={\rm Aut}(C)$, and hence the number of linear codes equivalent to $C$ is $\frac{|\mathbb{P}_n|}{{\rm Aut}(C)}=\frac{n!}{{\rm Aut}(C)}$.
Let $[C]$ denote the equivalence class of $C$ for any $C\in {\rm HHull}[n,k]_{q^2}^\ell$. Let
$$\widetilde{{\rm HHull}}[n,k]^\ell_{q^2}=\{[C]~|~C\in {\rm HHull}[n,k]_{q^2}^\ell\}.$$
Then $|\widetilde{{\rm HHull}}[n,k]^\ell_{q^2}|$ denotes the number of inequivalent linear $[n,k]_{q^2}$ codes with $\ell$-dimensional Hermitian hull.
Therefore, we have
$$\sum_{[C]\in \widetilde{{\rm HHull}}[n,k]^\ell_{q^2}}\frac{n!}{{\rm Aut}(C)}=|{\rm HHull}[n,k]_{q^2}^\ell|.$$
Note that the right-hand side of the above equation has been determined in Theorem \ref{Thm-mass-Hhull}. Hence, this provides a detailed explanation that the mass formula is an important tool for classifying linear codes. Finally, we end this subsection with a brief example.

\begin{example}
Let $\F_4=\{0,1,\omega,\omega^2\}$ and let $C_{4,2,i}~(i=1,2,\ldots,7)$ be the quaternary linear $[4,2]$ code with generator matrix $(I_2~M_i)$, where $M_i$ is a $2\times 2$ matrix given by
$$\begin{array}{llll}
    M_1=\begin{pmatrix}
1&\omega^2\\
\omega&\omega
\end{pmatrix}, & M_2=\begin{pmatrix}
0&\omega^2\\
0&0
\end{pmatrix}, & M_3=\begin{pmatrix}
0&0\\
0&1
\end{pmatrix}, & M_4=\begin{pmatrix}
1&\omega^2\\
\omega^2&\omega^2
\end{pmatrix}, \\
    M_5=\begin{pmatrix}
\omega&\omega^2\\
\omega^2&\omega
\end{pmatrix}, & M_6=\begin{pmatrix}
1&\omega\\
1&\omega^2
\end{pmatrix}, & M_7=\begin{pmatrix}
\omega^2&\omega\\
1&\omega
\end{pmatrix}. &
  \end{array}
$$
We verified that these quaternary linear codes all have one-dimensional Hermitian hull and they are pairwise inequivalent (with respect to permutation equivalence) by MAGMA \cite{magma}. In addition,
$$|{\rm Aut}(C_{4,2,i})|=3,1,2,2,12,1,3~(i=1,2,\ldots,7).$$
Hence $\sum_{i=1}^7\frac{4!}{{\rm Aut}(C_{4,2,i})}=90.$ By Theorem \ref{Thm-mass-Hhull}, it can be checked that $|{\rm HHull}[4,2]_{q^2}^1|=90.$ Therefore, there are $7$ inequivalent quaternary linear $[4,2]$ codes with one-dimensional Hermitian hull.
\end{example}

\begin{remark}
Note that there exist more general equivalence relations between two linear codes, such as the monomial equivalence. However, monomially equivalent linear codes may have different Hermitian hulls (except for linear codes over $\F_4$, see \cite{CMTQP,Chen1,Chen2,QINP-luo}), so we only consider the permutation equivalence in this paper.
\end{remark}

\subsection{The asymptotic behavior}
We now investigate the asymptotic behavior of linear codes with prescribed Hermitian hull dimension. For $q > 1$, let $g_{q,n}$ be the number defined by $g_{q,n}=\prod_{i=1}^n\left(1-\frac{1}{q^i}\right).$ The sequence $g_{q,1}, g_{q,2}, \ldots$ is strictly decreasing and has positive limit, which is denoted by $g_{q,\infty}$.
We can rewrite the Gaussian binomial coefficient $\begin{bmatrix}
                     n \\
                     k
                   \end{bmatrix}_q$
as
\begin{align}\label{eq-n-k}
 \begin{bmatrix}
                     n \\
                     k
                   \end{bmatrix}_q&=q^{k(n-k)}\frac{g_{q,n}}{g_{q,k}\cdot g_{q,n-k}}.
\end{align}

\begin{theorem}\label{thm-asm-1}
Let $\ell, k$ and $n$ be three positive integers such that $\ell\leq k\leq n-\ell$. For fixed $\ell$, we have
$$\lim_{{\tiny\begin{array}{c}
          k\rightarrow \infty \\
          (n-k)\rightarrow\infty
        \end{array}}
}\frac{|{\rm HHull}[n,k]_{q^2}^\ell|}{\begin{bmatrix}
                             n \\
                             k \\
                           \end{bmatrix}_{q^2}}=\frac{q^\ell\cdot g_{q,\infty}\cdot g_{q^4,\infty}}{h_{q^2,\ell}\cdot g_{q^2,\infty}\cdot g_{q^2,\infty}},$$
where $h_{q,\ell}=\prod_{i=1}^{\ell}(q^i-1)$ and
$g_{q,n}=\prod_{i=1}^n\left(1-\frac{1}{q^i}\right).$ In particular, if $\ell=0$, then
$$\lim_{{\tiny\begin{array}{c}
          k\rightarrow \infty \\
          (n-k)\rightarrow\infty
        \end{array}}
}\frac{|{\rm HHull}[n,k]_{q^2}^0|}{\begin{bmatrix}
                             n \\
                             k \\
                           \end{bmatrix}_{q^2}}=\frac{g_{q,\infty}\cdot g_{q^4,\infty}}{g_{q^2,\infty}\cdot g_{q^2,\infty}}.$$
\end{theorem}

\begin{proof}
Suppose that $k_0=k-\ell$.
It is observed that
$$\begin{bmatrix}
                             n \\
                             k \\
                           \end{bmatrix}_{q^2}=
\left(\prod_{i=k_0+1}^{k}\frac{q^{2n-2i+2}-1}{q^{2i}-1}\right)\begin{bmatrix}
                             n \\
                             k_0 \\
                           \end{bmatrix}_{q^2}=
\left(\prod_{i=1}^{\ell}\frac{q^{2n-2i-2k_0+2}-1}{q^{2k_0+2i}-1}\right)\begin{bmatrix}
                             n \\
                             k_0 \\
                           \end{bmatrix}_{q^2}.$$
For $1\leq i\leq \ell$, it can be checked that
$$\lim_{{\tiny\begin{array}{c}
          k\rightarrow \infty \\
          (n-k)\rightarrow\infty
        \end{array}}
}\frac{A_{n,k_0,i}}{\frac{q^{2n-2i-2k_0+2}-1}{q^{2k_0+2i}-1}}=\lim_{{\tiny\begin{array}{c}
          k\rightarrow \infty \\
          (n-k)\rightarrow\infty
        \end{array}}
}\frac{B_{n,k_0,i}}{\frac{q^{2n-2i-2k_0+2}-1}{q^{2k_0+2i}-1}}=
\frac{q}{q^{2i}-1}.$$
By Theorem \ref{Thm-mass-Hhull}, we have
\begin{align*}
 \lim_{{\tiny\begin{array}{c}
          k\rightarrow \infty \\
          (n-k)\rightarrow\infty
        \end{array}}
}\frac{|{\rm HHull}[n,k]_{q^2}^\ell|}{\begin{bmatrix}
                             n \\
                             k \\
                           \end{bmatrix}_{q^2}}
&=\lim_{{\tiny\begin{array}{c}
          k\rightarrow \infty \\
          (n-k)\rightarrow\infty
        \end{array}}
}\frac{\left(\prod_{i=1}^\ell \frac{N_0^{q+1}(n-k_0-2i+2)-1} {q^{2k_0+2i}-q^{2k_0}}\right)|{\rm HLCD}[n,k_0]_{q^2}|} {\left(\prod_{i=1}^{\ell}\frac{q^{2n-2i-2k_0+2}-1}{q^{2k_0+2i}-1}\right)
\begin{bmatrix}
n \\
 k_0 \\
 \end{bmatrix}_{q^2}} \\
  & =
\frac{q^\ell}{h_{q^2,\ell}}\cdot \lim_{{\tiny\begin{array}{c}
          k_0\rightarrow \infty \\
          (n-k_0)\rightarrow\infty
        \end{array}}
}\frac{|{\rm HLCD}[n,k_0]_{q^2}|}{
\begin{bmatrix}
n \\
 k_0 \\
 \end{bmatrix}_{q^2}}.
\end{align*}
Next, we only consider the asymptotic behavior of Hermitian LCD codes. Note that
\begin{align*}
  \lim_{{\tiny\begin{array}{c}
          k_0\rightarrow \infty \\
          (n-k_0)\rightarrow\infty
        \end{array}}}\frac{|{\rm HLCD}[n,k_0]_{q^2}|}{
\begin{bmatrix}
n \\
 k_0 \\
 \end{bmatrix}_{q^2}}&=\lim_{{\tiny\begin{array}{c}
          k_0\rightarrow \infty \\
          (n-k_0)\rightarrow\infty
        \end{array}}}\frac{q^{k_0(n-k_0)}\prod_{i=1}^{k_0}\frac{q^{n-k_0+i}-(-1)^{n-k_0+i}} {q^i-(-1)^i}}{q^{2k_0(n-k_0)}\frac{g_{q^2,n}}{g_{q^2,k_0}\cdot g_{q^2,n-k_0}}}\\
&=\lim_{{\tiny\begin{array}{c}
          k_0\rightarrow \infty \\
          (n-k_0)\rightarrow\infty
        \end{array}}}\left(\frac{g_{q^2,\infty}}{q^{k_0(n-k_0)}}\right)\prod_{i=1}^{k_0}\frac{q^{n-k_0+i}}{q^i-(-1)^i}\\
&=\lim_{{\tiny\begin{array}{c}
k_0\rightarrow \infty \\
(n-k_0)\rightarrow\infty
\end{array}}}\frac{g_{q^2,\infty}}{\prod_{i=1}^{k_0}\left(1-\frac{(-1)^i}{q^i}\right)}\\
&=\frac{g_{q,\infty}\cdot g_{q^4,\infty}}{g_{q^2,\infty}\cdot g_{q^2,\infty}}
\end{align*}
This completes the proof.
\end{proof}

\section{A mass formula for linear codes with prescribed symplectic hull dimension}\label{sec-4}

\subsection{Characterization of symplectic self-orthogonal and LCD codes}
A matrix $M$ is called a {\em skew-symmetric matrix} if $A^T=-A$ and $A$ has zero diagonal. The following proposition gives the classification of skew-symmetric matrices over $\F_q$ under the equivalence relation $M\sim QMQ^T$ \cite{equ-2}, where $Q$ is nonsingular.

\begin{prop}{\rm\cite{equ-2}}\label{prop-rank}
Let $M$ be a skew-symmetric $k\times k$ matrix of rank $t$ with zero diagonal over $\F_q$. Then $t$ is even and there exists a nonsingular matrix $Q$ such that
  $$QMQ^T={\rm diag}(\underbrace{J_{2},J_{2},\ldots,J_{2}}_{t/2},0,\ldots,0).$$
\end{prop}

\begin{theorem}\label{thm-k-even}
Let $C$ be a linear $[2n,k]_q$ code. Then $C$ is symplectic LCD if and only if $k$ is even and there is a basis ${\bf c}_1,{\bf c}_1',\ldots,{\bf c}_{\frac{k}{2}},{\bf c}'_{\frac{k}{2}}$ of $C$ such that, for any $i,j\in \left\{1,2,\ldots,\frac{k}{2}\right\}$, the following conditions hold.
\begin{itemize}
 \setlength{\itemsep}{2pt}
\setlength{\parsep}{0pt}
\setlength{\parskip}{0pt}
  \item [{\rm (1)}]  $\langle \ccc_i,  \ccc'_j\rangle_{\rm s}=0$, for $i\neq j$;
  \item [{\rm (2)}] $\langle\ccc_i,  \ccc'_i\rangle_{\rm s}=1$ and $\langle\ccc'_i,  \ccc_i\rangle_{\rm s}=-1$ for $1\leq i\leq \frac{k}{2}$.
\end{itemize}
\end{theorem}

\begin{proof}
Assume that $k$ is even and there is a basis ${\bf c}_1,{\bf c}_1',\ldots,{\bf c}_{\frac{k}{2}},{\bf c}'_{\frac{k}{2}}$ of $C$ such that for any $i,j\in \left\{1,2,\ldots,\frac{k}{2}\right\}$, the conditions (1) and (2) hold.
Let $G$ be the generator matrix of $C$ whose rows form the basis ${\bf c}_1,{\bf c}_1',\ldots,{\bf c}_{\frac{k}{2}},{\bf c}'_{\frac{k}{2}}$.
Then $G\Omega_n G^T={\rm diag}(\underbrace{J_{2},J_{2},\ldots,J_{2}}_{k/2}).$ This implies that $G\Omega_n G^T$ is nonsingular. Therefore, $C$ is symplectic LCD by Proposition \ref{prop-hull}.

Conversely, assume that $C$ is symplectic LCD. Let $G$ be a generator matrix of $C$. Since $(G\Omega_n G^T)^T=G\Omega_n^T G^T=-G\Omega_n G^T$ and ${\bf x}\Omega_n {\bf x}^T=0$ for any ${\bf x}\in \F_q^{2n}$, $G\Omega_n G^T$ is a skew-symmetric matrix. By Proposition \ref{prop-hull}, we know that $G\Omega_n G^T$ is nonsingular. Then it follows from Proposition \ref{prop-rank} that $k$ is even and there is a nonsingular matrix $Q$ such that
$$QG\Omega_n G^TQ^T=(QG)\Omega_n (QG)^T={\rm diag}(\underbrace{J_{2},J_{2},\ldots,J_{2}}_{k/2}).$$
Let $G'=QG$. Then $G'$ is also a generator matrix of $C$. Let ${\bf c}_i$ be the
$(2i-1)$st row of the matrix $G'$ and let ${\bf c}'_i$ be the $2i$th row of the matrix $G'$ for $i\in \left\{1,2,\ldots,\frac{k}{2}\right\}$. Hence ${\bf c}_1,{\bf c}_1',\ldots,{\bf c}_{\frac{k}{2}},{\bf c}'_{\frac{k}{2}}$ is the desired
basis.
\end{proof}

\begin{remark}
A basis of $C$ satisfying the conditions in Theorem \ref{thm-k-even} is called a {\em symplectic basis}. Hence, a linear code $C$ is symplectic LCD if and only if $C$ has a symplectic basis.
\end{remark}

By Theorem \ref{thm-k-even}, we have the following proposition.

\begin{prop}\label{cor-k-l-even}
If there is a linear $[2n,k]_q$ code with $\ell$-dimensional symplectic hull, then $k-\ell$ is even.
\end{prop}

\begin{proof}
Let $C$ be a linear $[2n,k]_q$ code with $\ell$-dimensional symplectic hull. Let ${\bf c}_1,{\bf c}_2,\ldots,{\bf c}_k$ be a basis of $C$ such that ${\bf c}_1,\ldots,{\bf c}_\ell$ is a basis of ${\rm Hull_s}(C)$. Then we claim that the linear code $C_1$ generated by ${\bf c}_{\ell+1},\ldots,{\bf c}_k$ is a symplectic LCD $[2n,k-\ell]_q$ code.
It can be checked that
$$C_1\cap C^{\perp_{\rm s}}=(C_1\cap C)\cap C^{\perp_{\rm s}}=C_1\cap {\rm Hull_s}(C)=\{{\bf 0}\},$$
otherwise there exists nonzero codeword ${\bf c}\in C_1$ such that ${\bf c}\in {\rm Hull_s}(C)$, which is a contradiction.
For any ${\bf c}\in {\rm Hull_s}(C_1)$, it is easy to see that $\langle {\bf c},{\bf c}_i \rangle_{\rm s}=0$ for $1\leq i\leq k$. This implies that ${\bf c}\in C^{\perp_{\rm s}}$. Together with the fact that ${\bf c}\in {\rm Hull_s}(C_1)\subseteq C_1$, we have ${\bf c}={\bf 0}$ from $C_1\cap C^{\perp_{\rm s}}=\{{\bf 0}\}$. This completes the proof of the claim.
Hence $C_1$ is a symplectic LCD $[2n,k-\ell]_q$ code. By Theorem \ref{thm-k-even}, the result follows.
\end{proof}

\begin{theorem}\label{thm-equivalent}
Let $C$ be a linear $[2n,k]_q$ code. If $C$ is symplectic LCD, then $k$ is even and there is a generator matrix $G$ of $C$ and a generator matrix $H$ of $C^{\perp_{\rm s}}$ such that
$$\begin{pmatrix}
 G \\
 H
 \end{pmatrix}\Omega_n\begin{pmatrix}
  G \\
  H
 \end{pmatrix}^T={\rm diag}(\underbrace{J_{2},J_{2},\ldots,J_{2}}_{n}).$$
\end{theorem}

\begin{proof}
Note that $C^{\perp_{\rm s}}$ is also symplectic LCD since $C$ is symplectic LCD.
By Theorem \ref{thm-k-even}, $C$ has a symplectic basis ${\bf c}_1,{\bf c}_1',\ldots,{\bf c}_{\frac{k}{2}},{\bf c}'_{\frac{k}{2}}$ and $C^{\perp_{\rm s}}$ also has a symplectic basis ${\bf c}_{\frac{k}{2}+1},{\bf c}_{\frac{k}{2}+1}',\ldots,{\bf c}_{n},{\bf c}'_n$. Let $G$ be the generator matrix of $C$ whose $(2i-1)$st row is ${\bf c}_i$ and whose $2i$th row is ${\bf c}'_i$ for $1\leq i\leq \frac{k}{2}$. Let $H$ be the generator matrix of $C^{\perp_{\rm s}}$ whose $(2i-1)$st row is ${\bf c}_i$ and whose $2i$th row is ${\bf c}'_i$ for $\frac{k}{2}+1\leq i\leq n$. Hence
$$\begin{pmatrix}
 G \\
 H
 \end{pmatrix}\Omega_n\begin{pmatrix}
  G \\
  H
 \end{pmatrix}^T={\rm diag}(\underbrace{J_{2},J_{2},\ldots,J_{2}}_{n}).$$
This completes the proof.
\end{proof}

\subsection{A mass formula for symplectic LCD codes}
In this subsection, we obtain a mass formula for symplectic LCD codes. We consider the action of the symplectic group on the set of symplectic LCD codes. By Theorem \ref{thm-k-even}, we only need to consider cases where both the length and dimension are even. Let SLCD$[2n,2k]_q$ denote the set of all symplectic LCD $[2n,2k]_q$ codes.

Let ${\bf 0}_n$ be the zero vector of length $n$ and let $C_0$ denote the linear code with the following generator matrix:
\begin{align}\label{G0}
  G_0 &= \left(
\begin{array}{c|c}
{\bf e}_1  &   {\bf 0}_n  \\
{\bf 0}_n  &  {\bf e}_1  \\
\vdots   & \vdots  \\
{\bf e}_k  &   {\bf 0}_n  \\
{\bf 0}_n  &  {\bf e}_k  \\
\end{array}
\right).
\end{align}
It can be checked that $C_0$ is a symplectic LCD $[2n,2k]$ code and $C_0^{\perp_{\rm s}}$ has a generator matrix as follows:
\begin{align}\label{H0}
  H_0 &= \left(
\begin{array}{c|c}
{\bf e}_{k+1}  &   {\bf 0}_n  \\
{\bf 0}_n  &  {\bf e}_{k+1}  \\
\vdots   & \vdots  \\
{\bf e}_n  &   {\bf 0}_n  \\
{\bf 0}_n  &  {\bf e}_n  \\
\end{array}
\right).
\end{align}

\begin{theorem}\label{thm-C1C2}
Let $C_1$ and $C_2$ be two symplectic LCD $[2n,2k]_q$ codes. Then there exists a symplectic matrix $Q\in \mathbb{S}\mathrm{p}^1_{2n}(q)$ such that $C_2=C_1Q$. Conversely, for any symplectic LCD $[2n,2k]_q$ code $C$ and any symplectic matrix $Q\in \mathbb{S}\mathrm{p}^1_{2n}(q)$, $CQ$ is also a symplectic LCD $[2n,2k]_q$ code.
\end{theorem}

\begin{proof}
By Theorem \ref{thm-k-even}, $C_1$ has a symplectic basis ${\bf c}_1,{\bf c}_1',\ldots,{\bf c}_{\frac{k}{2}},{\bf c}'_{\frac{k}{2}}$ of and $C_1^{\perp_s}$ has a symplectic basis ${\bf c}_{\frac{k}{2}+1},{\bf c}_{\frac{k}{2}+1}',\ldots,{\bf c}_{n},{\bf c}'_n$. Then, ${\bf c}_1,{\bf c}_1',\ldots,{\bf c}_{n},{\bf c}'_{n}$ is a symplectic basis of $C_1\oplus C_1^{\perp_s}=\F_q^{2n}$. Similarly, there exists another symplectic basis ${\bf w}_1,{\bf w}_1',\ldots,{\bf w}_{n},{\bf w}'_{n}$ of $C_2\oplus C_2^{\perp_s}=\F_q^{2n}$ such that ${\bf w}_1,{\bf w}_1',\ldots,{\bf w}_{k},{\bf w}'_{k}$ is a symplectic basis of $C_2$ and ${\bf w}_{k+1},{\bf w}_{k+1}',\ldots,{\bf w}_{n},{\bf w}'_{n}$ is a symplectic basis of $C_2^{\perp_s}$. Let $Q_1$ and $Q_2$ be the matrices defined by
$$
Q_1=\begin{pmatrix}
{\bf c}_1    \\
{\bf c}'_{1}    \\
\vdots     \\
{\bf c}_n    \\
{\bf c}'_n    \\
\end{pmatrix}~{\rm and}~Q_2=
\begin{pmatrix}
{\bf w}_1    \\
{\bf w}'_{1}    \\
\vdots     \\
{\bf w}_n    \\
{\bf w}'_n    \\
\end{pmatrix}.
$$
Then we have
$$Q_1\Omega_nQ_1^T=Q_2\Omega_nQ_2^T={\rm diag}(\underbrace{J_{2},J_{2},\ldots,J_{2}}_{n}).$$
By Proposition \ref{prop-rank}, there exists a nonsingular matrix $A$ such that
$$A\Omega_nA^T={\rm diag}(\underbrace{J_{2},J_{2},\ldots,J_{2}}_{n}).$$
Hence we have
$$A^{-1}Q_1\Omega_nQ_1^TA^{-T}=A^{-1}Q_2\Omega_nQ_2^TA^{-T}=A^{-1}{\rm diag}(\underbrace{J_{2},J_{2},\ldots,J_{2}}_{n})A^{-T}=\Omega_n.$$
Hence $A^{-1}Q_1,A^{-1}Q_2\in \mathbb{S}\mathrm{p}^1_{2n}(q)$. Let $Q=Q_1^{-1}Q_2$. Then
$Q=Q_1^{-1}Q_2=Q_1^{-1}AA^{-1}Q_2=(A^{-1}Q_1)^{-1}(A^{-1}Q_2)\in \mathbb{S}\mathrm{p}^1_{2n}(q)$ is a symplectic matrix with respect to $\Omega_n$. Further, it is easy to see that ${\bf c}_iQ={\bf w}_i$ and ${\bf c}'_iQ={\bf w}'_i$ for $1\leq i\leq k$. Thus, $C_1Q=C_2$.

Conversely, let $C$ be a symplectic LCD $[2n,2k]_q$ code and let $Q\in \mathbb{S}\mathrm{p}^1_{2n}(q)$ be a symplectic matrix with respect to $\Omega_n$. Let $G$ be a generator matrix of $C$. Then $GQ$ is a generator matrix of $CQ$. Note that
$(GQ)\Omega_n(GQ)^T=G(Q\Omega_nQ^T)G^T=G\Omega_nG^T.$
This implies that $CQ$ is also a symplectic LCD $[2n,2k]_q$ code.
\end{proof}

Theorem \ref{thm-C1C2} shows that the symplectic group $\mathbb{S}\mathrm{p}^1_{2n}(q)$ acts on SLCD$[2n,2k]$ by $(C,Q)\longmapsto CQ$, where $C\in {\rm SLCD}[2n,2k]_q$ and $Q\in \mathbb{S}\mathrm{p}^1_{2n}(q)$. Under this action, all symplectic LCD codes are on a unique orbit. Note that the code $C_0$ generated by the matrix $G_0$ defined in Equation (\ref{G0}) is a symplectic LCD $[2n,2k]_{q}$ code. Hence we have
\begin{align}\label{Eq-sLCD}
  {\rm SLCD}[2n,2k]_q&=C_0\mathbb{S}\mathrm{p}^1_{2n}(q).
\end{align}
Hence, to determine the size of ${\rm SLCD}[2n,2k]_q$, it is sufficient to determine the size of the stabilizer of $C_0$ in the symplectic group $\mathbb{S}\mathrm{p}^1_{2n}(q)$. The stabilizer of a linear $[2n,k]_q$ code $C$ in the symplectic group $\mathbb{S}\mathrm{p}^1_{2n}(q)$ is defined by
$$\mathbb{S}{\rm t}_s(C)=\{Q\in \mathbb{S}\mathrm{p}^1_{2n}(q)~|~CQ=C\}.$$

\begin{prop}\label{lem-St}
Let $C$ be a symplectic LCD $[2n,2k]_q$ code. Let $G$ be a generator matrix of $C$ and let $H$ be a generator matrix of $C^{\perp_s}$. Then $Q\in \mathbb{S}{\rm t}_s(C)$ if and only if
$$Q=\begin{pmatrix}
 G \\
 H
 \end{pmatrix}^{-1}
\begin{pmatrix}
 Q_1&O \\
 O& Q_2
 \end{pmatrix}
\begin{pmatrix}
 G \\
 H
 \end{pmatrix},$$
 where $Q_1\in \mathbb{G}\mathbb{L}_{2k}(q)$, $Q_2\in \mathbb{G}\mathbb{L}_{2n-2k}(q)$ such that $Q_1(G\Omega_nG^T)Q_1^T=G\Omega_nG^T$ and $Q_2(H\Omega_nH^T)Q_2^T=H\Omega_nH^T$.
\end{prop}

\begin{proof}
Let $Q=\begin{pmatrix}
 G \\
 H
 \end{pmatrix}^{-1}
\begin{pmatrix}
 Q_1&O \\
 O& Q_2
 \end{pmatrix}
\begin{pmatrix}
 G \\
 H
 \end{pmatrix},$
 where $Q_1\in \mathbb{G}\mathbb{L}_{2k}(q)$, $Q_2\in \mathbb{G}\mathbb{L}_{2n-2k}(q)$ such that $Q_1(G\Omega_nG^T)Q_1^T=G\Omega_nG^T$ and $Q_2(H\Omega_nH^T)Q_2^T=H\Omega_nH^T$. Then
 $$\begin{pmatrix}
 G \\
 H
 \end{pmatrix}Q=
\begin{pmatrix}
 Q_1&O \\
 O& Q_2
 \end{pmatrix}
 \begin{pmatrix}
 G \\
 H
 \end{pmatrix}=\begin{pmatrix}
 Q_1G \\
 Q_2H
 \end{pmatrix}.$$
 Hence $GQ=Q_1G$ and $HQ=Q_2H$, which implies that $CQ=C$. Note that
\begin{align*}
\begin{pmatrix}
 G \\
 H
 \end{pmatrix}Q\Omega_nQ^T\begin{pmatrix}
 G \\
 H
 \end{pmatrix}^T & =\left(\begin{pmatrix}
 G \\
 H
 \end{pmatrix}Q\right)\Omega_n\left(\begin{pmatrix}
 G \\
 H
 \end{pmatrix}Q\right)^T \\
   & =\begin{pmatrix}
 Q_1G \\
 Q_2H
 \end{pmatrix}\Omega_n\begin{pmatrix}
 Q_1G \\
 Q_2H
 \end{pmatrix}^T\\
 &=\begin{pmatrix}
 Q_1(G\Omega_nG^T)Q_1^T &O\\
 O& Q_2(H\Omega_nH^T)Q_2^T
 \end{pmatrix}\\
 &=\begin{pmatrix}
G\Omega_nG^T &O\\
 O& H\Omega_nH^T
 \end{pmatrix}\\
 &=\begin{pmatrix}
 G \\
 H
 \end{pmatrix}\Omega_n\begin{pmatrix}
 G \\
 H
 \end{pmatrix}^T.
\end{align*}
It follows that
$$Q\Omega_nQ^T=\begin{pmatrix}
 G \\
 H
 \end{pmatrix}^{-1}\left(\begin{pmatrix}
 G \\
 H
 \end{pmatrix}\Omega_n\begin{pmatrix}
 G \\
 H
 \end{pmatrix}^T\right)\begin{pmatrix}
 G \\
 H
 \end{pmatrix}^{-T}=\Omega_n,$$
which implies that $Q\in \mathbb{S}\mathrm{p}^1_{2n}(q)$. Thus, $Q\in \mathbb{S}{\rm t}_s(C)$.

Conversely, let $Q\in \mathbb{S}{\rm t}_s(C)$, i.e., $Q\in \mathbb{S}\mathrm{p}^1_{2n}(q)$ and $CQ=C$. Then there exists a matrix $Q_1\in \mathbb{G}\mathbb{L}_{2n}(q)$ such that
$GQ=Q_1G$. For any ${\bf c}\in C$ and any ${\bf w}\in C^{\perp_s}$, we have
$$\langle{\bf c}Q,{\bf w}Q\rangle_{\rm s}={\bf c}Q\Omega_n({\bf w}Q)^T={\bf c}(Q\Omega_nQ^T){\bf w}^T={\bf c}\Omega_n{\bf w}^T=\langle{\bf c},{\bf w}\rangle_{\rm s}.$$
Then $C^{\perp_{\rm s}}Q=C^{\perp_{\rm s}}$. Hence there exists a matrix $Q_2\in \mathbb{G}\mathbb{L}_{2n}(q)$ such that $HQ=Q_2H$. It follows that
$$\begin{pmatrix}
 G \\
 H
 \end{pmatrix}Q=
\begin{pmatrix}
 Q_1&O \\
 O& Q_2
 \end{pmatrix}
\begin{pmatrix}
 G \\
 H
 \end{pmatrix},~{\rm i.e.},~Q=\begin{pmatrix}
 G \\
 H
 \end{pmatrix}^{-1}
\begin{pmatrix}
 Q_1&O \\
 O& Q_2
 \end{pmatrix}
\begin{pmatrix}
 G \\
 H
 \end{pmatrix}.$$
Since $Q\in \mathbb{S}{\rm p}^1_{2n}(q)$, $Q\Omega_nQ=\Omega_n$ and we have
\begin{align*}
  G\Omega_nG^T&=G(Q\Omega_nQ^T)G^T=(GQ)\Omega_n(GQ)^T=
(Q_1G)\Omega_n(Q_1G)^T=Q_1(G\Omega_nG^T)Q_1^T,\\
  H\Omega_nH^T&=H(Q\Omega_nQ^T)H^T=(HQ)\Omega_n(HQ)^T=
(Q_2H)\Omega_n(Q_2H)^T=Q_2(H\Omega_nH^T)Q_2^T.
\end{align*}
This completes the proof.
\end{proof}

\begin{theorem}\label{thm-mass}
Assume that $n$ and $k$ are two positive integers with $0<k<n$. Then the number of distinct symplectic LCD $[2n,2k]$ codes is
$$|{\rm SLCD}[2n,2k]|=q^{2k(n-k)}\begin{bmatrix}
                     n \\
                     k
                   \end{bmatrix}_{q^2}.$$
\end{theorem}

\begin{proof}
Let $C_0$ be the symplectic LCD $[2n,2k]_{q}$ code with the generator matrix $G_0$ defined in Equation (\ref{G0}). Then the matrix $H_0$ defined in Equation (\ref{H0}) is a generator matrix of $C_0^{\perp_{\rm s}}$. Note that
$$G_0\Omega_nG_0^T={\rm diag}(\underbrace{J_{2},J_{2},\ldots,J_{2}}_{k})~
{\rm and}~
H_0\Omega_nH_0^T={\rm diag}(\underbrace{J_{2},J_{2},\ldots,J_{2}}_{n-k}).$$
Together with Proposition \ref{lem-St}, we have
$$\mathbb{S}{\rm t}_s(C_0)=\left\{\begin{pmatrix}
 G_0 \\
 H_0
 \end{pmatrix}^{-1}
\begin{pmatrix}
 Q_1&O \\
 O& Q_2
 \end{pmatrix}
 \begin{pmatrix}
 G_0 \\
 H_0
 \end{pmatrix}~\bigg|~Q_1\in \mathbb{S}{\rm p}^2_{2k}(q), Q_2\in \mathbb{S}{\rm p}^2_{2n-2k}(q)\right\}.$$
From Equation (\ref{Eq-sLCD}), we have
$$|{\rm SLCD}[2n,2k]_q|=\frac{|\mathbb{S}\mathrm{p}^1_{2n}(q)|}{|\mathbb{S}{\rm t}_s(C_0)|}=\frac{|\mathbb{S}\mathrm{p}^1_{2n}(q)|}
{|\mathbb{S}\mathrm{p}^2_{2k}(q)|\cdot |\mathbb{S}\mathrm{p}^2_{2n-2k}(q)|}.$$
We obtain the desired result by Equation $(\ref{eq-sp})$.
\end{proof}

\subsection{A mass formula for linear codes with prescribed symplectic hull dimension}

Jin and Xing \cite{ISIT-Jin} presented a mass formula for symplectic self-orthogonal codes in order to obtain the Gilbert-Varshamov bound for QECCs. In this subsection, we obtain a mass formula for linear codes with prescribed symplectic hull dimension. Let ${\rm SHull}[2n,k]_q^{\ell}$ denote the set of all linear $[2n,k]_q$ codes with $\ell$-dimensional symplectic hull. By Proposition \ref{cor-k-l-even}, $k-\ell$ is even.

\begin{theorem}\label{thm-sym-mass-hull}
Assume that $n$, $k$ and $\ell$ are three nonnegative integers with $\ell\leq k\leq 2n$ and $k=2k_0+\ell$. Then
$$|{\rm SHull}[2n,k]_q^{\ell}|=\left(\prod_{i=1}^{\ell} \frac{q^{2n-k-i+2}-1}{q^{k}-q^{2k_0}}\right) q^{2k_0(n-k_0)}\begin{bmatrix}
                     n \\
                     k_0
                   \end{bmatrix}_{q^2}.$$
\end{theorem}

\begin{proof}
We use $C^{\ell}_{2n,k}$ to denote a linear $[2n,k]_q$ code with $\ell$-dimensional symplectic hull. Then for $k<n$, a given
$C^{\ell-1}_{2n,k-1}$ can be extended to a $C^{\ell}_{2n,k}$ by adding a vector ${\bf c}\in (C^{\ell-1}_{2n,k-1})^{\perp_{\rm s}}\backslash {\rm Hull_s}(C^{\ell-1}_{2n,k-1})$.
Note that every vector in $\F_q^{2n}$ is orthogonal to itself with the symplectic inner product.
Thus, we can obtain $q^{2n-(k-1)}-q^{\ell-1}$ linear $[2n,k]_q$ codes with $\ell$-dimensional symplectic hull from $C^{\ell-1}_{2n,k-1}$ through this way. Note that some of these codes may be the same.
On the other hand, for any $C^{\ell}_{2n,k}$, there are $(q^{\ell}-1)q^{k-1}$ ways to obtain $C^{\ell}_{2n,k}$ through the above way by a similar discussion as \cite[Proposition 5.5]{2023IT-hull}. Thus, we obtain
\begin{align*}
 |{\rm SHull}[2n,k]_q^{\ell}| & =\frac{q^{2n-(k-1)}-q^{\ell-1}}{(q^{\ell}-1)q^{k-1}}\cdot |{\rm SHull}[2n,k-1]_q^{\ell-1}| \\
   & =\frac{q^{2n-2k_0-2\ell+2}-1}{q^{2k_0+\ell}-q^{2k_0}}\cdot |{\rm SHull}[2n,k-1]_q^{\ell-1}|.
\end{align*}
Repeating the above operation $\ell$ times, we obtain
$$|{\rm SHull}[2n,k]_q^{\ell}|=\left(\prod_{i=1}^{\ell} \frac{q^{2n-2k_0-2i+2}-1}{q^{2k_0+i}-q^{2k_0}}\right) \cdot |{\rm SHull}[2n,2k_0]^0_q|.$$
Note that ${\rm SHull}[2n,k_0]^0_{q}={\rm SLCD}[2n,k_0]_{q}$. Hence the result holds by Theorem \ref{thm-mass}.
\end{proof}

In order to obtain the Gilbert-Varshamov bound for QECCs, Jin and Xing \cite{ISIT-Jin} gave a formula for the number of distinct symplectic self-orthogonal codes. Unfortunately,  there is a small typo in the formula. Later, Jin {\em et al.} \cite{IT-Jin} gave a corrected mass formula for symplectic self-orthogonal codes. Next, we give the mass formula for symplectic self-orthogonal codes, which can serve as a special case of Theorem \ref{thm-sym-mass-hull}.

\begin{cor}{\rm\cite[Lemma 3.5]{IT-Jin}}\label{Cor-SSO}
Assume that $n$ and $k$ are two integers with $1\leq k\leq n$. Then the number of distinct symplectic self-orthogonal $[2n,k]_q$ codes is
$$|{\rm SHull}[2n,k]_q^k|=\frac{(q^{2n-2k+2}-1)(q^{2n-2k+4}-1)\cdots(q^{2n}-1)}
{(q^k-1)(q^{k-1}-1)\cdots(q-1)}.$$
\end{cor}

\begin{proof}
Note that a symplectic self-orthogonal $[2n,k]_q$ code is a linear $[2n,k]_q$ code with $k$-dimensional symplectic hull. Then $k_0=0$. By
Theorem \ref{thm-sym-mass-hull}, we have
$$|{\rm SHull}[2n,k]_q^k|=\left(\prod_{i=1}^{k} \frac{q^{2n-2i+2}-1}{q^{i}-1}\right)=\frac{(q^{2n-2k+2}-1)(q^{2n-2k+4}-1)\cdots(q^{2n}-1)}
{(q^k-1)(q^{k-1}-1)\cdots(q-1)}.$$
This completes the proof.
\end{proof}

In order to obtain the Gilbert-Varshamov bound for QECCs, they \cite{ISIT-Jin} also gave a formula for the number of distinct symplectic self-orthogonal codes containing a given nonzero vector. Unfortunately, we also find that there is a small typo in the formula. Next, we give a corrected formula.

\begin{lem}\label{lem-SSO}
Assume that $n$ and $k$ are two integers with $1\leq k\leq n$. Given a nonzero vector ${\bf u}\in \F_q^{2n}$, the number of distinct symplectic self-orthogonal $[2n,k]_q$ codes containing ${\bf u}$ is
$$\frac{(q^{2n-2k+2}-1)(q^{2n-2k+4}-1)\cdots(q^{2n-2}-1)}
{(q^{k-1}-1)\cdots(q-1)}.$$
\end{lem}

\begin{proof}
We denote by $C_{2n,k}({\bf u})$ a symplectic self-orthogonal $[2n,k]_q$ code containing ${\bf u}$ and denote by $B_k({\bf u})$ the number of such codes. Then a given $C_{2n,k-1}({\bf u})$ can be extended to some $C_{2n,k}({\bf u})$ by adding a vector ${\bf c}\in (C_{2n,k-1})^{\perp_{\rm s}}\backslash {\rm Hull_s}(C_{2n,k-1})$.
Note that every vector in $\F_q^{2n}$ is orthogonal to itself with symplectic inner product.
Thus, we can obtain $q^{2n-(k-1)}-q^{k-1}$ symplectic self-orthogonal $[2n,k]_q$ codes containing ${\bf u}$ from $C_{2n,k-1}({\bf u})$ through this way. Note that some of these codes may be the same.

For any $C_{2n,k}({\bf u})$, we claim that there are $(q^{k-1}-1)q^{k-1}$ ways to obtain $C_{2n,k}({\bf u})$ through the above way.
It is obvious that there is a subcode $C_{2n,k-1}({\bf u})$ of $C_{2n,k}({\bf u})$ and a vector ${\bf x}$ such that
$$C_{2n,k}({\bf u})=C_{2n,k-1}({\bf u})\oplus \langle {\bf x}\rangle.$$
Now it is sufficient to count the number of ${\bf x}$'s and $C_{2n,k-1}({\bf u})$'s that satisfy the above equation.
By \cite[Theorem 4, p.698]{MacWilliams}, the number of $C_{2n,k-1}({\bf u})$'s is
$N=\begin{bmatrix}
k-1 \\
k-2 \\
\end{bmatrix}_q=\frac{q^{k-1}-1}{q-1}.$
For a fixed $C_{2n,k-1}({\bf u})$, the number of ${\bf x}$'s is $N_{\bf x}=|C_{2n,k}({\bf u})\backslash C_{2n,k-1}({\bf u})|=q^k-q^{k-1}$.
It follows that there are $N\cdot N_{\bf x}=(q^{k-1}-1)q^{k-1}$ ways to obtain $C_{2n,k}({\bf u})$ through the above way. Hence the claim is valid.

Thus, we obtain
\begin{align*}
 B_k({\bf u}) & =\frac{(q^{2n-(k-1)}-q^{k-1})\cdot B_{k-1}({\bf u})}{(q^{k-1}-1)q^{k-1}}=\frac{q^{2n-2(k-1)}-1}{q^{k-1}-1}\cdot B_{k-1}({\bf u}).
\end{align*}
Repeating the above operation $k-1$ times, we obtain
$$B_k({\bf u})=\left(\prod_{i=1}^{k-1}\frac{q^{2n-2i}-1}{q^{i}-1}\right) \cdot B_1({\bf u}).$$
Combining with the fact that $B_1({\bf u})=1$, the result follows.
\end{proof}

Finally, we end this subsection with a remark.
\begin{remark}
By exhaustive search \cite{magma}, we find that there are $40$ distinct ternary symplectic self-orthogonal $[4,2]$ codes, and four of them are symplectic self-orthogonal $[4,2]$ codes containing the codeword $(1,0,0,0)$. This is precisely confirmed through the formulas obtained in Corollary \ref{Cor-SSO} and Lemma \ref{lem-SSO}.

In fact, a linear $[2n,k]_q$ code is equivalent to an additive $(n,q^k)_{q^2}$ code. There is also a correspondence between inner products. We do not introduce them here and one can refer to \cite{IT-qc,quantum-codes-IT-2,FFA-Sym,Dcc-EAqecc} for this topic. With arguments similar to these in the Hermitian case, we can see that
the mass formula in Theorem \ref{thm-sym-mass-hull} is useful for classifying inequivalent additive codes over $\F_{q^2}$ with prescribed hull dimension.
\end{remark}

\subsection{The asymptotic behavior}

In the following, we will analyze the asymptotic behavior of the size of the orbits $C_0\mathbb{S}\mathrm{p}^1_{2n}(q)$. Furthermore, we investigate the asymptotic behavior of linear codes with prescribed symplectic hull dimension. Recall that the sequence $g_{q,1}, g_{q,2}, \ldots$ is strictly decreasing and has positive limit, which is denoted by $g_{q,\infty}$.

\begin{theorem}
Assume that $n$, $k$ and $\ell$ are three nonnegative integers with $\ell\leq k\leq 2n$ and $k=2k_0+\ell$.
Then for fixed $\ell$, we have
$$\lim_{{\tiny\begin{array}{c}
          k\rightarrow \infty \\
          (2n-k)\rightarrow\infty
        \end{array}}
}\frac{|{\rm SHull}[2n,k]_q^\ell|}{\begin{bmatrix}
                             2n \\
                             k \\
                           \end{bmatrix}_q}
=\frac{q^{\ell}\cdot g_{q,\infty}}{h_{q,\ell}\cdot g_{q^2,\infty}},$$
where $h_{q,\ell}=\prod_{i=1}^{\ell}(q^i-1)$ and $g_{q,n}=\prod_{i=1}^n\left(1-\frac{1}{q^i}\right).$
In particular, if $\ell=0$, then
$$\lim_{{\tiny\begin{array}{c}
          k\rightarrow \infty \\
          (2n-k)\rightarrow\infty
        \end{array}}
}\frac{|{\rm SHull}[2n,k]^0_q|}{\begin{bmatrix}
                             2n \\
                             k \\
                           \end{bmatrix}_q}=
\lim_{{\tiny\begin{array}{c}
          k_0\rightarrow \infty \\
          (n-k_0)\rightarrow\infty
        \end{array}}
}\frac{|{\rm SLCD}[2n,2k_0]_q|}{\begin{bmatrix}
                             2n \\
                             2k_0 \\
                           \end{bmatrix}_q}
=\frac{g_{q,\infty}}{g_{q^2,\infty}}.$$
\end{theorem}

\begin{proof}
From Theorem \ref{thm-mass} and Equation $(\ref{eq-n-k})$,
\begin{align*}
  \lim_{{\tiny\begin{array}{c}
          k_0\rightarrow \infty \\
          (n-k_0)\rightarrow\infty
        \end{array}}
}\frac{|{\rm SLCD}[2n,2k_0]_q|}{\begin{bmatrix}
                             2n \\
                             2k_0 \\
                           \end{bmatrix}_q} &=\lim_{{\tiny\begin{array}{c}
          k_0\rightarrow \infty \\
          (n-k_0)\rightarrow\infty
        \end{array}}
}\frac{q^{2k_0(n-k_0)}\begin{bmatrix}
                     n \\
                     k_0
                   \end{bmatrix}_{q^2}}{\begin{bmatrix}
                             2n \\
                             2k_0 \\
                           \end{bmatrix}_q} \\
   &= \lim_{{\tiny\begin{array}{c}
          k_0\rightarrow \infty \\
          (n-k_0)\rightarrow\infty
        \end{array}}
}\frac{q^{2k_0(n-k_0)}\cdot q^{2k_0(n-k_0)}\frac{g_{q^2,n}}{g_{q^2,k_0}\cdot g_{q^2,n-k_0}}}{q^{2k_0(2n-2k_0)}\frac{g_{q,2n}}{g_{q,2k_0}\cdot g_{q,2n-2k_0}}}\\
   &=\lim_{{\tiny\begin{array}{c}
          k_0\rightarrow \infty \\
          (n-k_0)\rightarrow\infty
        \end{array}}
}\frac{g_{q^2,n}\cdot g_{q,2k_0}\cdot g_{q,2n-2k_0}}{g_{q,2n}\cdot g_{q^2,k_0}\cdot g_{q^2,n-k_0}}\\
&=\frac{g_{q,\infty}}{g_{q^2,\infty}}.
\end{align*}
This completes the proof of the latter part. Next, let us consider the first part.
It is observed that
{\small  $$\begin{bmatrix}
2n \\
  k \\
  \end{bmatrix}_q  =\frac{q^{2n-k+1}-1}{q^k-1}\begin{bmatrix}
                             2n \\
                             k-1 \\
                           \end{bmatrix}_q \\
    =\left(\prod_{i=2k_0+1}^{k}\frac{q^{2n-i+1}-1}{q^i-1}\right)\begin{bmatrix}
                             2n \\
                             2k_0 \\
                           \end{bmatrix}_q\\
   =
\left(\prod_{i=1}^{\ell}\frac{q^{2n-i-2k_0+1}-1}{q^{2k_0+i}-1}\right)\begin{bmatrix}
                             2n \\
                             2k_0 \\
                           \end{bmatrix}_q.$$}
For $1\leq i\leq \ell$, it can be checked that
\begin{align*}
  \lim_{{\tiny\begin{array}{c}
          k\rightarrow \infty \\
          (2n-k)\rightarrow\infty
        \end{array}}
}\frac{\frac{q^{2n-k-i+2}-1}{q^{k}-q^{2k_0}}}{\frac{q^{2n-i-2k_0+1}-1}{q^{2k_0+i}-1}} &  =\frac{q}{q^i-1}.
\end{align*}
From Theorem \ref{thm-sym-mass-hull} and the above argument, we have
\begin{align*}
  \lim_{{\tiny\begin{array}{c}
          k\rightarrow \infty \\
          (2n-k)\rightarrow\infty
        \end{array}}
}\frac{|{\rm SHull}[2n,k]_q^\ell|}{\begin{bmatrix}
                             2n \\
                             k \\
                           \end{bmatrix}_q} & =  \lim_{{\tiny\begin{array}{c}
          k_0\rightarrow \infty \\
          (n-k_0)\rightarrow\infty
        \end{array}}
} \frac{\left(\prod_{i=1}^{\ell}  \frac{q^{2n-k-i+2}-1}{q^{k}-q^{2k_0}}\right)|{\rm SLCD}[2n,2k_0]_q|}
   {\left(\prod_{i=1}^{\ell}\frac{q^{2n-i-2k_0+1}-1}{q^{2k_0+i}-1}\right)\begin{bmatrix}
                             2n \\
                             2k_0 \\
                           \end{bmatrix}_q}\\
   &=\frac{q^{\ell}\cdot g_{q,\infty}}{h_{q,\ell}\cdot g_{q^2,\infty}}.
\end{align*}
This completes the proof.
\end{proof}

\section{Conclusion}\label{sec-5}
In this paper, we have pushed further the study of mass formulas for linear codes with prescribed Hermitian and symplectic hull dimensions. We have obtained mass formulas for Hermitian (resp. symplectic) LCD codes by considering the action of the unitary (resp. symplectic) group on the set of all Hermitian (resp. symplectic) LCD codes. As a consequence, we have obtained mass formulas for linear codes with prescribed Hermitian and symplectic hull dimensions, and have further studied their asymptotic behavior. In addition, we have also provided the corrections for the formulas obtained by Jin and Xing \cite{ISIT-Jin} for symplectic self-orthogonal codes.\\




\end{sloppypar}

\end{document}